\documentclass[12pt,onecolumn,draft]{IEEEtran}
\usepackage{amsmath,amstext,amsthm,amssymb,epsf}

\newtheorem{theorem}{\sc Theorem}

\newtheorem{lemma}{\sc Lemma}
\newtheorem{coro}{\sc Corollary}
\newtheorem{req}{\sc Requirement}
\newtheorem{nota}{\sc Notation}
\newtheorem{defin}{\sc Definition}
\newtheorem{rem}{\sc Remark}
\newtheorem{cla}{\sc Claim}
\newtheorem{ex}{\sc Example}

\newenvironment{remark}{\begin{rem}}{\hspace*{\fill}$\Diamond$\end{rem}}

\newenvironment{corollary}{\begin{coro}}{\end{coro}}

\newenvironment{definition}{\begin{defin}}{\end{defin}}

\title{Conditional Kolmogorov Complexity and Universal Probability}
\author{Paul M.B. Vit\'anyi\thanks{
CWI, Science Park,
1098 XG Amsterdam, The Netherlands. 
Email: {\tt Paul.Vitanyi@cwi.nl}.}
}

\date{}
\begin{document}
\maketitle

\begin{abstract}
The Coding Theorem of L.A. Levin connects 
unconditional prefix Kolmogorov complexity with
the discrete universal distribution. 
There are conditional versions referred to in several publications but as yet
there exist no written proofs in English. 
Here we provide those proofs. They use a different definition than the
standard one for the conditional
version of the discrete universal distribution.
Under the classic definition
of conditional probability, there is no conditional version of the Coding Theorem. 
\end{abstract}

\section{Introduction}
Informally, the Kolmogorov complexity, or algorithmic entropy, of a
string $x$ is the length (number of bits) of a shortest binary
program (string) to compute
$x$ on a fixed reference universal computer
(such as a particular universal Turing machine).
Intuitively, this quantity represents the minimal amount of information
required to generate $x$ by any effective process.
The conditional Kolmogorov complexity of $x$ relative to
$y$ is defined similarly as the length of a shortest binary program
to compute $x$, if $y$ is furnished as an auxiliary input to the
computation \cite{Ko65}.

The Coding Theorem \eqref{eq.codth} of L.A. Levin \cite{Le74} 
connects a variant of Kolmogorov complexity,
the unconditional prefix Kolmogorov complexity, with
the discrete universal distribution. 
The negative logarithm of the latter
is up to a constant equal to the former. 
The conditional in conditional Kolmogorov complexity commonly is taken
to be a finite binary string. 

A conditional version of the Coding Theorem as 
referred to in \cite{Ga74,Le76,LV08,Ga10,SUV13} 
requires a function denoted as ${\bf m}(x|y)$ with $x,y \in \{0,1\}^*$ 
that is (i) lower semicomputable;
(ii) satisfies  $\sum_x {\bf m}(x|y) \leq 1$ for every $y$; (iii) if $p(x|y)$ is a
function satisfying (i) and (ii) then there is a constant $c$
such that $c {\bf m}(x|y) \geq p(x|y)$ for all $x$ and $y$.
There is no written complete proof of the conditional version of the Coding Theorem.
Our aim is to provide such a proof and write it out in detail rather than
rely on ``clearly'' or ``obviously.'' One wants to be certain that 
applications of the conditional version of the Coding Theorem 
are well founded. 

Since
the discrete universal distribution ${\bf m}$ over one variable is a semiprobability mass
function, that is $\sum_x {\bf m}(x) \leq 1$, it is
natural to consider a universal distribution ${\bf m}(x,y)$ over two variables
with $\sum_{x,y}{\bf m}(x,y) \leq 1$. One then can
define the conditional version following the custom in probability
theory, for example \cite{Sh48}, 
\begin{equation}\label{eq.mxy1}
{\bf m} (x|y) = \frac{{\bf m}(x,y)}{\sum_z {\bf m}(z,y)}.
\end{equation}
But in \cite{Ga74,Le76,LV08,Ga10,SUV13}
the conditional semiprobability 
${\bf m} (x|y)$ is defined differently, namely as 
in Definition~\ref{def.mxy2}. 
In Theorem~\ref{theo.single} for a single distribution, and in 
Theorem~\ref{theo.nocoding} for a joint distribution, it is shown that if one uses
\eqref{eq.mxy1} then ${\bf m} (x|y)$ does not satisfy a Coding
Theorem. In contrast, if ${\bf m} (x|y)$ is defined according
to Definition~\ref{def.mxy2} it does have a Coding Theorem, Theorem~\ref{PR2}. 

The necessary notions
and concepts are given in Appendices: 
Appendix~\ref{app.codes} introduces
prefix codes, Appendix~\ref{app.kolm} introduces
Kolmogorov complexity, Appendix~\ref{app.complexity} introduces complexity 
notions, and Appendix~\ref{app.precision} tells about our use of $O(1)$.

\subsection{related work}
We can enumerate all lower semicomputable probability mass functions
with one argument. 
For convenience these arguments are elements of $\{0,1\}^*$.
The enumeration list is denoted
\[
{\cal P} = P_1, P_2, \ldots .
\] 
There is another interpretation possible. Let 
prefix Turing machine $T_i$ be the $i$th element in the standard
enumeration of prefix Turing machines $T_1, T_2, \ldots .$
Then $R_i(x)= \sum 2^{-|p|}$ where $p$ is a program for $T_i$ such that
$T_i(p)=x$. This $R_i(x)$
is the probability that prefix Turing machine $T_i$ outputs $x$
when the program on its input tape is supplied by flips of a fair coin.
We can thus form the list
\[
{\cal R} = R_1, R_2, \ldots .
\]
Both lists ${\cal P}$ and ${\cal R}$ enumerate the same functions and there are
computable isomorphisms between the two \cite{LV08} Lemma 4.3.4.
\begin{definition}
\rm
If $U$ is the reference universal prefix Turing machine, then the corresponding
distribution in the $R$-list is $R_U$.
\end{definition}

L.A. Levin \cite{Le74} proved that
\begin{equation}\label{eq.muni}
{\bf m} (x) = \sum_j \alpha_j  P_j (x),
\end{equation}
with $\sum_j \alpha_j \leq 1$, $\alpha_j > 0$, and $\alpha_j$ lower semicomputable, 
is a universal lower semicomputable semiprobability mass function. 
(For semiprobabilities see Appendix~\ref{app.complexity}.) That is,
obviously it is lower semicomputable and $\sum_x {\bf m} (x) \leq 1$. 
It is called a {\em universal} lower semicomputable semiprobability mass function
since (i) it is itself a lower semicomputable semiprobability mass function and
(ii) it multiplicatively (with factor $ \alpha_j$) dominates every 
lower semicomputable semiprobability mass function $P_j$.

Moreover, he proved the Coding Theorem
\begin{equation}\label{eq.codth}
- \log {\bf m}(x) = - \log R_U(x) =  K(x),
\end{equation}
where equality holds up to a constant additive term.

\subsection{Results} 
We give a review of the classical definition of conditional probability
versus the one used in the case of semicomputable probability. 
In Sections~\ref{sect.single} and \ref{sect.joint} we show that 
the conditional version of \eqref{eq.codth} do not hold for
the classic definition of conditional probability in the case of a
single probability distribution
(Theorem~\ref{theo.single}) and for joint
distributions (Theorem~\ref{theo.nocoding}).
In Section~\ref{sect.cond} we consider the
Definition~\ref{def.mxy2} 
of the conditional version of joint semicomputable semiprobability mass functions
as used in \cite{Ga74,Le76,LV08,Ga10,SUV13}. For this definition the conditional
version of \eqref{eq.codth} holds.
We write all proofs out in complete detail. 

\section{Preliminaries}\label{sect.prel}
Let $x,y,z \in {\cal N}$, where
${\cal N}$ denotes the natural
numbers and we identify
${\cal N}$ and $\{0,1\}^*$ according to the
correspondence
\[(0, \epsilon ), (1,0), (2,1), (3,00), (4,01), \ldots \]
Here $\epsilon$ denotes the {\em empty word}.
A {\em string} $x$ is an element of $\{0,1\}^*$.
The {\em length} $|x|$ of $x$ is the number of bits
in $x$, not to be confused with 
the absolute value of a number. Thus,
$|010|=3$ and $|\epsilon|=0$, while 
$|-3|=|3|=3$.

The emphasis is on binary sequences only for convenience;
observations in any alphabet can be so encoded in a way
that is `theory neutral'.
Below we will use the natural numbers and the binary strings
interchangeably.

\section{Conditional Probability}\label{sect.single}
Let $P$ be a probability mass function on sample space ${\cal N}$, that is, 
$\sum P(x)=1$ where the summation is over ${\cal N}$. Suppose we consider 
$x \in {\cal N}$ and event $B \subseteq {\cal N}$ has occurred. 
According to Kolmogorov in \cite{Ko33} 
a new probability $P(x|B)$
has arisen satisfying:
\begin{enumerate}
\item
$x \not\in B$: $P(x|B)=0$;
\item 
$x \in B$: $P(x|B) =  P(x)/P(B)$;
\item
$\sum_{x \in B} P(x|B)= 1$.
\end{enumerate}
Let ${\bf m}$ be as defined in \eqref{eq.muni} with the sample space ${\cal N}$. 
Then $\sum {\bf m}(x) \leq 1$ and we call 
${\bf m}$ a {\em semiprobability}. For the conditional versions of
semiprobabilities Items 1) an 2) 
above hold and Item 3) holds with $\leq$. 
We show that in with these definitions there is no conditional Coding Theorem.
\begin{theorem}\label{theo.single}
Let $B \subseteq {\cal N}$ and $|B| \leq \infty$. 
Then 
$- \log {\bf m}(x|B) \neq K(x|B) +O(1)$. 
\end{theorem}
\begin{proof}
($x \not\in B$) 
This implies 
${\bf m}(x|B) = 0$ and therefore $- \log {\bf m}(x|B) = \infty$. 
But $K(x|B) < \infty$.

($x \in B$) 
We can replace $B$ by its
characteristic string: $|\chi_{B}|=|B|$ and $\chi_{B}$ is defined by
$\chi_{B}(i)=1$ if $i \in B$ and $\chi_{B}(i)=0$ otherwise. 
Rewrite the conditional 
\[
{\bf m}(x|B) = 
 \frac{{\bf m}(x)}{{\bf m}(B)}
 =\frac{{\bf m}(x)}{{\bf m}(\chi_{B})}.
\]
Then, applying the Coding Theorem on the single argument numerator and denominatir of
the right-hand side,
\[
-\log {\bf m}(x|B) 
= K(x) - K(\chi_{B}).
\]
Let $K(\chi_{B}) \geq |B|$. For every $x \in B$ we have
$K(x) \leq \log |B| + O(\log \log |B|)$. Then, $-\log {\bf m}(x|B) \leq -|B|/2$. 
But for every $x$ and $B$ we have $K(x|B) \geq 0$.
\end{proof}

\section{Lower Semicomputable Joint Probability}\label{sect.joint}
We show that there is no equivalent of the Coding Theorem for 
the conditional version of ${\bf m}$ according to \eqref{eq.mxy1}
 based on lower semicomputable 
joint probability mass functions.
We use a standard pairing function $\langle \cdot , \cdot \rangle$
to obtain two-argument (joint) lower semicomputable probability
mass functions from the single argument ones. For example,
$\langle i , j \rangle = \frac{1}{2} (i+j)(i+j+1)+j$.

\begin{definition}
\rm
Let $x,y \in {\cal N}$ and
$f(\langle x,y \rangle )$ be a lower semicomputable function on a single
argument such that 
we have $\sum_{\langle x,y \rangle} f(\langle x,y \rangle) \leq 1$. 
We use these functions $f$ to define the
{\em lower semicomputable joint semiprobability mass
functions} $Q_j(x,y) = f(\langle x,y \rangle )$.
\end{definition}
Let us define the list
\[
{\cal Q}=Q_1,Q_2, \ldots . 
\]
We can effectively enumerate
the family of lower semicomputable joint semiprobability mass functions
as before by ${\cal Q}$.
We can now define the {\em lower semicomputable joint universal probability} by
\begin{equation}\label{eq.mxy}
{\bf m} (x,y) = \sum_j  \alpha_j  Q_j ( x,y),
\end{equation}
with $\sum_j \alpha_j \leq 1$. Classically,
for a joint probability mass function $P(x,y)$ with $x,y \in {\cal N}$
and $\sum_{x,y} P(x,y)=1$ one defines the conditional version \cite{CT91} by
\[
P(x|y) = \frac{P(x,y)}{\sum_z P(z,y)}.
\]
We call $P_1(x)=\sum_z P(x,z)$ and  $P_2(y)=\sum_z P(z,y)$ the 
{\em marginal} probability of $x$ and $y$, respectively.
This form of conditional $P(x|y)$ corresponds with $P(x|B)$ 
in Section~\ref{sect.single}  in that $B=\{(z,y):z \in {\cal N}\}$.
The semiprobability ${\bf m}$ in \eqref{eq.mxy1} satisfies $\sum_{x,y}{\bf m}(x,y) \leq 1$
and the analogue of the above yields
\begin{definition}
\label{def.mxy3}
\rm
The {\em conditional} version of ${\bf m}(x,y)$ is
defined by
\begin{align*}
{\bf m} (x|y)& = \frac{{\bf m}(x,y)}{\sum_z {\bf m}(z,y)}
\\&= \frac{ \sum_j \alpha_j Q_j (x,y)}
{\sum_z\sum_j \alpha_j Q_j (z,y)}
\\& =\frac{\sum_j \alpha_j   Q_j (x,y)}
{\sum_j \alpha_j \sum_z Q_j (z,y)}
\end{align*}
\end{definition}

This {\em conditional} version ${\bf m}(x|y)$ is the quotient of 
two lower semicomputable functions. 
It may not be
semicomputable (not proved here).
We show that there is
no conditional coding theorem for this version of ${\bf m}(x|y)$.

\begin{theorem}\label{theo.nocoding}
Let $x,y \in {\cal N}$. Then, 
$- \log {\bf m}(x|y) \geq K(x|y) +  O(|y|)$.
The $O(|y|)$ term in general cannot be improved.
\end{theorem}

\begin{proof}
By \eqref{eq.mxy} and the Coding Theorem 
we have $-\log {\bf m}(x,y) = K(\langle x,y \rangle)+O(1)$. 
Clearly, $K(\langle x,y \rangle) = K(x,y)+O(1)$.
The marginal universal probability ${\bf m}_2(y)$ is given
by ${\bf m}_2(y)= \sum_z {\bf m}(z,y) \geq {\bf m}(\epsilon , y)$. 
Thus, with the last equality due to the Coding Theorem:
$-\log {\bf m}_2(y) \leq -\log {\bf m}(\epsilon , y)
=K(\langle \epsilon,y \rangle)+O(1)=K(y)+O(1)$.
By the Symmetry of Information \eqref{eq.soi} we find
$K(x,y)=K(y)+K(x|y,K(y))+O(1)$.  Here $K(x|y,K(y))=K(x|y)+O(\log |y|)$. Since
${\bf m}(x|y) = {\bf m}(x,y)/{\bf m}_2(y)$ by Definition~\ref{def.mxy2},
we have $- \log {\bf m}(x|y) = - \log {\bf m}(x,y) + \log {\bf m}_2(y) \geq 
 - \log {\bf m}(x,y) + \log {\bf m}(\langle \epsilon,y \rangle) =
K(x|y) + O (\log(|y|)$. Here the first inequality follows
from the relation between ${\bf m}_2(y)$ and ${\bf m}(\langle \epsilon,y\rangle)$, while the last
equality follows from \eqref{eq.soi}.
In \cite{Ga74} it is shown that for every length of 
the binary representation of $y \in {\cal N}$ 
there are $y$ such
that $K(x|y,K(y)) = K(x|y)+ \Omega(\log |y|)$.
\end{proof}

\section{Lower Semicomputable Conditional Probablity}\label{sect.cond}
We consider lower semicomputable conditional semiprobabilities directly
in order to obtain 
a conditional semiprobability that (i) is lower semicomputable itself,
and (ii) dominates multiplicatively every lower semicomputable
conditional semiprobability.
Let $f(x,y)$ be a lower semicomputable function.
We use these functions $f$
to define
{\em lower semicomputable conditional semiprobability mass
functions} $P(x|y)$.

\begin{theorem}
There is a universal conditional 
lower semicomputable
semiprobability mass function. We denote it by ${\bf m}$.
\label{PR1}
\end{theorem}
\begin{proof}
We prove the theorem in two steps. In Stage 1 we show that the
two-argument lower semicomputable functions which sum over the first argument
to at most 1 can be effectively enumerated as
$$
P_1 , P_2 ,  \ldots \;.
$$
This enumeration contains all and only lower 
semicomputable conditional semiprobability
mass functions.
In Stage 2 we show that $P_0$ as defined below multiplicatively
dominates all $P_j$:
$$
P_0 (x|y) = \sum_j \alpha_j P_j (x|y),
$$
with $\sum \alpha_j \leq  1$, and $\alpha_j > 0$ and lower semicomputable.
Stage 1 consists of two parts. In the first part, we
enumerate all lower semicomputable two argument functions; 
and in the second part
we effectively change the lower semicomputable two argument functions
to functions that sum to at most 1 over the first argument.
Such functions 
leave the functions that were
already conditional lower semicomputable   
semiprobability mass functions unchanged.

{\sc Stage 1}
Let $\psi_1 , \psi_2 , \ldots $ be an effective enumeration
of all two-argument real-valued partial recursive functions. For example,
let $\psi_1 (x,y) , \psi_2 (x,y) , \ldots $ be $\psi_1 (\langle x,y \rangle) , 
\psi_2 (\langle x,y \rangle) , \ldots $ with $\langle \cdot, \cdot \rangle$
the standard pairing function over the natural numbers. Consider
a function $\psi$ from this enumeration
(where we drop the subscript for
notational convenience). Without loss of generality,
assume that each $\psi$ is approximated
by a rational-valued three-argument partial recursive function
$\phi ' (x,y, k) = p/q$
(use
$\phi ' (\langle \langle  x, y \rangle , k \rangle ) = \langle p, q \rangle $).
Without loss of generality, each such $\phi'$ is
modified to a partial recursive function
satisfying the properties below.
For all $x,y,k  \in  {\cal N}$,
\begin{itemize}
\item
if $\phi  (x,y, k) < \infty$,
then also $\phi (x,y, 1) , \phi (x,y, 2) , \ldots  , \phi (x,y, k-1) < \infty$
(this can be achieved by the trick of dovetailing the
computation of $\phi ' (\langle \langle x,y\rangle, 1 \rangle)$, 
$\phi ' (\langle \langle x,y\rangle,2 \rangle)$, $\ldots$
and assigning computed values in enumeration order of halting to
$\phi  (x,y, 1) , \phi (x,y, 2) , \ldots$);
\item
$\phi  (x,y, k + 1)  \geq  \phi (x,y, k)$ (dovetail the
computation of $\phi ' (x,y,1)$, $\phi ' (x,y,2)$, $\ldots$
and assign the enumerated values
to $\phi  (x,y, 1) , \phi (x,y, 2) , \ldots$ satisfying this requirement
and ignoring the other computed values); and
\item
$\lim_{{k}  \rightarrow  \infty} \phi (x,y, k) = \psi (x,y)$
(as does $\phi'$).
\end{itemize}
The resulting $\psi$-list contains all
lower semicomputable two-argument real-valued functions, and is
represented by the approximators in the
$\phi$-list.
Each lower semicomputable function $\psi$ (rather, the
approximating function $\phi$) will be used to construct
a function $P$ that sums to at most 1 over the first argument.
In the algorithm below, the local
variable array
$P$ contains the current
approximation to the values of
$P$ at each stage of the computation.
This is doable because the nonzero part of the approximation
is always finite.
\begin{description}
\item[{\bf Step 1:}]$\;\;\;$
Initialize by setting
$P (x|y) := 0$ for all $x,y  \in  {\cal N}$; and set $k := 0$.
\item[{\bf Step 2:}]$\;\;\;$
Set $k := k + 1$, and compute $\phi (1,1, k) , \ldots , \phi (k,k, k)$.
{\bf \{}If any $\phi (i,j, k)$, $1  \leq  i,j  \leq  k$, is undefined, then
the existing values of $P$ do not change.{\bf \}} 
\item[{\bf Step 3:}]$\;\;$
{\bf if} for some $j$ ($1 \leq j \leq k$) we have 
$\phi (1,j, k) + \cdots + \phi (k,j, k)  >  1$  {\bf then}
the existing values of $P$ do not change {\bf else}
{\bf for} $i,j:=1, \ldots,k$ set $P (i|j) := \phi (i,j, k)$ 
{\bf \{}Step 3 is a test of
whether the new assignment of $P$-values satisfy (also future)
lower semicomputable conditional semiprobability 
mass function requirements{\bf \}}
\item[{\bf Step 4:}]$\;\;\;$
{\bf Go to} Step 2.
\end{description}
If $\psi (x,y)$ satisfies $\sum_x \psi(x,y) \leq 1$ for all $x,y \in {\cal N}$ 
then $P(x|y) = \psi(x,y)$ for all $x,y \in {\cal N}$.
If for some $x,y$ and $k$ with $x,y  \leq  k$ the value $\phi( x,y, k)$ is
undefined, then the last assigned values of
$P$ do not change any more
even though the computation
goes on forever. If the {\bf else} condition in Step 3 is 
satisfied in the limit with equality by the values of $P$, 
it is a conditional semiprobability mass function.
If {\bf if} condition in Step 3 gets satisfied, then the computation
terminates and $P$'s support is finite and it is computable. 

Executing this procedure on all functions
in the list $\phi_1 , \phi_2 , \ldots $ yields an effective
enumeration $P_1 ,  P_2 , \ldots $ of lower semicomputable functions 
containing all and only
lower semicomputable conditional semiprobability mass functions.
The algorithm takes care that for all $j \geq 1$ we have
\[
\sum_x P_j(x|y) \leq 1.
\]
{\sc Stage 2}
Define the function $P_0$ as
$$
P_0 (x|y)  = \sum_j \alpha_j   P_j (x|y),
$$
with $\alpha_j$ chosen such that
$\sum_j  \alpha_j \leq  1$, $\alpha_j > 0$ and
lower semicomputable for all $j$.
Then $P_0$ is a conditional semiprobability mass function since
\begin{eqnarray*}
\sum_x P_0 (x|y)  =  \sum_j \alpha_j 
\sum_x  P_j (x|y) 
 \leq  \sum_j  \alpha_j  
\leq 1.
\end{eqnarray*}
The function $P_0(\cdot|\cdot)$ is also lower semicomputable,
since $P_j (x|y)$ is
lower semicomputable in $j$ and $x,y$. (Use the universal
partial recursive function $\phi_0$ and the construction above.)
Also $\alpha_j$ is by definition lower semicomputable for all $j$.
Finally, $P_0$ multiplicatively dominates each $P_j$ 
since for all $x,y \in {\cal N}$ we have
$P_0 (x|y)  \geq  \alpha_j  P_j (x|y)$ while $\alpha_j >0$. 
Therefore, $P_0$ is a universal lower semicomputable conditional 
semiprobability mass function.
\end{proof}
We can choose the $\alpha_j$'s in the definition of $P_0$ in the  proof above by setting 
\[
\alpha_j= 2^{- K(j)-c_j} ,
\]
with the $c_j \geq 0$ constants.
Then $\sum_j  \alpha_j  \leq  1$ by the ubiquitous 
Kraft inequality \cite{Kr49} (satisfied by the 
prefix complexity $K$), and $\alpha_j> 0$ and
lower semicomputable for all $j$.
\begin{definition}
\label{def.mxy2}
\rm
We define 
$$
{\bf m} (x|y) = \sum_{{j}  \geq  1}  2^{{-} K(j)-c_j}  P_j (x|y).
$$
We call ${\bf m}(x|y)$ the {\em reference universal lower semicomputable conditional
semiprobability mass function}.
\end{definition}
\begin{corollary}\label{PR1corollary}
\rm
If $P(x|y)$
is a lower semicomputable conditional semiprobability mass function, then
$2^{K(P)} {\bf m}(x|y) \geq P(x|y)$, for all $x,y$.
That is, ${\bf m}(x|y)$
multiplicatively dominates every lower semicomputable conditional 
semiprobability mass function $P(x|y)$.
\end{corollary}

\subsection{A Priori Probability}

\label{sect.discrete.prior.probability}
Let $P_1 , P_2 , \ldots $ be the effective enumeration
of all lower semicomputable conditional 
semiprobability mass functions constructed
in Theorem~\ref{PR1}. There is another way to
effectively enumerate all lower semicomputable conditional
semiprobability mass functions.
Let the input to a
prefix machine $T$ (with the string $y$ on its auxiliary tape)
be provided
by an infinitely long sequence of fair coin flips.
The probability of generating an initial input segment $p$
is $2^{-|p|}$.
If $T(p,y) < \infty$, that is, $T$'s computation on $p$ with $y$ on its 
auxiliary tape
terminates, then presented with any infinitely
long sequence starting with
$p$, the machine $T$ with $y$ on its auxiliary tape, being a prefix machine,
will read exactly $p$ and no further.

Let $T_1 , T_2 , \ldots $ be the standard
enumeration of prefix machines in \cite{LV08}.
For each
prefix machine $T$, define
\begin{equation}\label{eq.univprob}
Q_T (x|y) = \sum_{{T(p,y)} = x}  2^{-|p|}.
\end{equation}
In other words, $Q_T (x|y)$ is the probability that $T$
with $y$ on its auxiliary tape computes output $x$ if its input is provided by
successive tosses of a fair coin.
This means that for every string $y$ we have that $Q_T$ satisfies
$$
\sum_{{x}  \in  {\cal N}}  Q_T (x|y)  \leq  1.
$$
We can approximate $Q_T(\cdot|y)$ for every string $y$  as follows. 
(The algorithm
uses the local variable $Q (x)$ to store the
current approximation to $Q_T (x|y)$.)
\begin{description}
\item[{\bf Step 1:}] $\;\;\;$
Fix $y \in \{0,1\}^*$. Initialize $Q (x) := 0$ for all $x$.
\item[{\bf Step 2:}] $\;\;\;$
Dovetail the running of all programs on $T$ with auxiliary $y$ so
that in stage $k$, step $k-j$ of program $j$ is executed.
Every time the computation of some program $p$
halts with output $x$, increment
$Q (x) := Q (x)+2^{-|p|}$.
\end{description}
The algorithm approximates the
displayed sum in Equation~\ref{eq.univprob}
by the contents of $Q (x)$. Since $Q(x)$ is nondecreasing,
this shows that $Q_T$ is lower semicomputable.
Starting from a standard enumeration of prefix machines
$T_1 ,T_2 ,\ldots$, this construction gives for every $y \in \{0,1\}^*$
an enumeration of only lower semicomputable conditional probability mass functions
$$
Q_1 (\cdot|y) , Q_2 (\cdot|y),  \ldots \;. 
$$
To merge the enumerations for different $y$ we use dovetailing over the index $i$ of
$Q_i$ and $y$.
The $P$-enumeration of Theorem~\ref{PR1} contains
all elements enumerated by this $Q$-enumeration.
In \cite{LV08} Lemma 4.3.4 the reverse is shown.
\begin{definition}\label{def.uapp}
\rm
The 
\it conditional universal a priori
probability
\rm
on the positive integers is defined as
$$
Q_U (x|y) = \sum_{{U(p,y)} = x}  2^{-|p|} ,
$$
where $U$ is the reference prefix machine.
\end{definition}
\begin{remark}
\rm
The use of prefix machines in the present discussion rather
than plain Turing machines is necessary.
By Kraft's inequality
the series $\sum_p  2^{-|p|}$ converges (to $\leq 1$) if the
summation is taken over all halting
programs $p$ of any fixed prefix machine with a fixed auxiliary input $y$.
In contrast, if
the summation is taken over all halting programs $p$
of a universal plain Turing machine, then
the series $\sum_p 2^{-|p|}$ diverges.
\end{remark}

\subsection{The Conditional Coding Theorem}

\begin{theorem}\label{PR2}
There is a constant $c$ such that for every $x$,
$$
 \log \frac{1}{ {\bf m} (x|y)} =  \log \frac{1}{Q_U (x|y)} = K(x|y),
$$
with equality up to an additive constant $c$.
\end{theorem}
\begin{proof}
Since $2^{-K(x|y)}$ represents the contribution to
$Q_U (x|y)$ by a shortest program for $x$ given the auxiliary $y$,
we have $2^{-K(x|y)}  \leq  Q_U (x|y)$, for all $x,y$.

Clearly, $Q_U (x|y)$ is lower semicomputable. Namely, 
enumerate all programs for $x$ given $y$, by running reference
machine $U$ on all programs with $y$ as auxiliary at once in dovetail fashion:
in the first phase, execute step 1 of program 1; in the second phase,
execute step 2 of program 1 and step 1 of program 2;
in the $i$th phase ($i>2$), execute step $j$ of program $k$ for all
positive $j$ and $k$ such that $j + k  = i$.
By the universality of ${\bf m} (x|y)$ in the class
of lower semicomputable conditional semiprobability mass functions,
$Q_U (x|y)  = O(  {\bf m} (x|y))$.

It remains to show that
$ {\bf m} (x|y) = O( 2^{-K(x|y)} )$.
This is equivalent to proving that
$K(x|y) \leq  \log 1/ {\bf m}  (x|y) + O(1)$, as follows.
Exhibit a prefix-code $E$ encoding each
source word $x$ given $y$ as a code word $E(x|y) = p$, satisfying
$$
|p|  \leq    \log \frac{ 1}{ {\bf m} (x|y)} + O(1),
$$
together with a decoding prefix machine $T$ such that
$T(p,y) = x$. Then,
$$
K_T (x|y)  \leq  |p|, 
$$
and by the Invariance Theorem \eqref{eq.it} 
$$
K(x|y)  \leq  K_T (x|y) + c_T,
$$
with $c_T >0$ a constant that may depend on $T$ but not on $x,y$.
Note that $T$ is fixed by the above construction.
On the way to constructing $E$ as required, we recall a construction for
the Shannon--Fano code:
\begin{lemma}\label{claim.shannon-fano}
If $p$ is a function
on the nonnegative integers, 
and $\sum_x p (x)  \leq  1$, then there is a binary prefix-code $e$
such that the code words $e(1),  e(2), \ldots $ can be
length-increasing lexicographically ordered
and $|e(x)|  \leq   \log 1/p (x) + 2$. 
\end{lemma}
\begin{proof}
Let $[0, 1)$ be the half-open real unit interval,
corresponding to the sample space $S =  \{ 0,1 \}^{\infty}$.
Each element $\omega$ of $S$ corresponds to a real
number $0. \omega$.
Let $x  \in   \{ 0, 1 \}^*$.
The half-open interval $[0.x, 0.x + 2^{-|x|} )$ corresponding
to the cylinder (set) of reals
$\Gamma_x =  \{ 0. \omega : \omega = x \ldots    \in  S \} $
is called a %
\it binary interval
\rm .
We cut off disjoint, consecutive,
adjacent (not necessarily binary) intervals
$I_x$ of length $p (x)$ from the left end of $[0, 1)$,
$x = 1, 2, \ldots \;. $ Let $i_x$ be the length of the longest
binary interval contained in $I_x$. Set $E(x)$ equal
to the binary word corresponding to the leftmost such interval.
Then
$|e(x)| =  \log 1/ i_x$.
It is easy to see that $I_x$ is covered by at most four
binary intervals of length $i_x$, from which the lemma follows.
\end{proof}
We use this construction to find
a prefix machine $T$ such that
$K_T (x|y)  \leq   \log 1/ {\bf m} (x|y) + c$.
That ${\bf m} (x|y)$ is not computable but
only lower semicomputable results in $c = 3$.

Since ${\bf m} (x|y)$ is lower semicomputable,
there is a partial recursive function $\phi (x,y,t)$
with $\phi (x,y,t)  \leq  {\bf m} (x|y)$ and
$\phi (x,y,t+1)  \geq  \phi (x,y,t)$, for all $t$.
Moreover, $\lim_{{t}  \rightarrow  \infty} \phi (x,y,t) = {\bf m} (x|y)$.
Let $\psi (x,y,t)$
be the greatest partial recursive lower bound of the following
special form on $\phi (x,y,t)$ defined by
$$
\psi (x,y,t) :=
 \{ 2^{-k}: 2^{-k} \leq  \phi (x,y,t) < 2 \cdot 2^{-k}
\mbox{ and } \phi (x,y,j) < 2^{-k} \mbox{ for all } j < t\},
$$
and $\psi (x,y,t) := 0$ otherwise.
Let $\psi$ enumerate its range without
repetition. Then,
\begin{eqnarray*}
\sum_{x,y,t}   \psi (x,y,t)   =  
\sum_x \sum_y   \sum_t   \psi (x,y,t) 
 \leq  2 {\bf m} (x|y) 
 \leq   2.
\end{eqnarray*}
The series $\sum_{x,y,t}  \psi (x,y,t)$ can converge
to precisely $2 {\bf m} (x|y)$ only in case there is
a positive integer $k$ such that ${\bf m} (x|y) = 2^{-k}$.

In a manner similar to the above proof
we chop off consecutive, adjacent, disjoint
half-open intervals $I_{x,y,t}$ of length $\psi (x,y,t)/2$,
in enumeration order of a dovetailed computation
of all $\psi (x,y,t)$, starting
from the left-hand side of $[0, 1)$. We have already shown that this
is possible. It is easy to see that we can
construct a prefix machine $T$
as follows: If $\Gamma_p$ is the leftmost largest binary interval
of $I_{x,y,t}$, then $T(p,y) = x$. Otherwise, $T(p,y) = \infty$
($T$ does not halt).

By construction of $\psi$, for each pair $x,y$ there is a $t$ such that
$\psi (x,y,t) > {\bf m} (x|y)/2$.
Each interval $I_{x,y,t}$ has length
$\psi (x,y,t)/2$. Each $I$-interval contains a binary interval $\Gamma_p$
of length at least one-half of that of $I$ (because the length of $I$
is of the form $2^{-k}$, it contains a binary interval
of length $2^{-k-1}$) . Therefore,
there is a $p$ with $T(p,y) = x$ such that $2^{-|p|}  \geq  {\bf m} (x|y)/8$.
This implies $K_T (x|y)  \leq   \log 1/ {\bf m} (x|y) + 3$, which was
what we had to prove.
\end{proof}
\begin{corollary}
\rm
The above result plus Corollary~\ref{PR1corollary} give:
If $P$ is a lower semicomputable conditional semiprobability mass function.
Then there is a constant $c_{P} = K (P ) + O(1)$
such that $K(x|y)  \leq   \log 1/ P(x|y) + c_{P}$. 
\end{corollary}

\section{Conclusion}
The conditional version of the Coding Theorem 
of L.A. Levin, Theorem~\ref{PR2}, requires
a lower semicomputable conditional semiprobability that multiplicatively
dominates all other lower semicomputable conditional semiprobabilities as in
Theorem~\ref{PR1}. The conventional form of the conditional \eqref{eq.mxy1},
applied to the distribution \eqref{eq.muni} satisfying
the original Coding Theorem \eqref{eq.codth} is false. This is shown by
Theorems~\ref{theo.single} and \ref{theo.nocoding}.

\appendix
\subsection{Self-delimiting Code}\label{app.codes}
A binary string $y$
is a {\em proper prefix} of a binary string $x$
if we can write $x=yz$ for $z \neq \epsilon$.
 A set $\{x,y, \ldots \} \subseteq \{0,1\}^*$
is {\em prefix-free} if for any pair of distinct
elements in the set neither is a proper prefix of the other.
A prefix-free set is also called a {\em prefix code} and its
elements are called {\em code words}.
An example of a
prefix code, that is useful later,
encodes the source word $x=x_1 x_2 \ldots x_n$ by the code word
\[ \overline{x} = 1^n 0 x .\]
This prefix-free code
is called {\em self-delimiting}, because there is fixed computer program
associated with this code that can determine where the
code word $\bar x$ ends by reading it from left to right without
backing up. This way a composite code message can be parsed
in its constituent code words in one pass, by the computer program.
(This desirable property holds for every prefix-free
encoding of a finite set of source words, but not for every
prefix-free encoding of an infinite set of source words. For a single
finite computer program to be able to parse a code message the encoding needs
to have a certain uniformity property like the $\overline{x}$ code.)
Since we use the natural numbers and the binary strings interchangeably,
$|\bar x|$ where $x$ is ostensibly an integer, means the length
in bits of the self-delimiting code of the binary string with index $x$.
On the other hand, $\overline{|x|}$ where $x$ is ostensibly a binary
string, means the self-delimiting code of the binary string
with index the length $|x|$ of $x$.
Using this code we define
the standard self-delimiting code for $x$ to be
$x'=\overline{|x|}x$. It is easy to check that
$|\overline{x} | = 2 n+1$ and $|x'|=n+2 \log n +1$.
Let $\langle \cdot \rangle$ denote a standard invertible
effective one-one encoding from ${\cal N} \times {\cal N}$
to a subset of ${\cal N}$.
For example, we can set $\langle x,y \rangle = x'y$.
We can iterate this process to define
$\langle x , \langle y,z \rangle \rangle$,
and so on. For Kolmogorov complexity it is essential that there
exists a pairing function such that
the length of $\langle u,v \rangle$ is equal to the sum of
the lengths of $u,v$ plus a small value depending only on $|u|$.)

\subsection{Kolmogorov Complexity}\label{app.kolm}
For  precise definitions, notation, and results see the text \cite{LV08}.
For technical reasons we use a variant of complexity,
so-called prefix complexity, which is associated with Turing machines
for which the set of programs resulting in a halting computation
is prefix free.
We realize prefix complexity by considering a special type of Turing
machine with a one-way input tape, a separate work tape,
and a one-way output tape. Such Turing
machines are called {\em prefix} Turing machines. If a machine $T$ halts
with output $x$
after having scanned all of $p$ on the input tape,
but not further, then $T(p)=x$ and
we call $p$ a {\em program} for $T$.
It is easy to see that
$\{p : T(p)=x, x \in \{0,1\}^*\}$ is a {\em prefix code}.

Let $T_1 ,T_2 , \ldots$ be a standard enumeration
of all prefix Turing machines with a binary input tape,
for example the lexicographical length-increasing ordered 
prefix Turing machine descriptions \cite{LV08}.
Let $\phi_1 , \phi_2 , \ldots$
be the enumeration of corresponding prefix functions
that are computed by the respective prefix Turing machines
($T_i$ computes $\phi_i$).
These functions are the {\em partial recursive} functions
or {\em computable} functions (of effectively prefix-free encoded
arguments). 
We denote the function computed by a Turing machine $T_i$ with $p$ as input
and $y$ as conditional information by
$\phi_i(p,y)$.
One of the main achievements of the theory of computation
is that the enumeration $T_1,T_2, \ldots$ contains
a machine, say $T_u$, that is computationally universal and optimal in that it can
simulate the computation of every machine in the enumeration when
provided with its program and index. Namely, it computes a
function $\phi_u$ such that
   $\phi_u(\langle i, p\rangle,y)  = \phi_i (p,y)$
    for all $i,p,y$.
    We fix one such machine and designate it as the {\em reference universal
    Turing machine} or {\em reference Turing machine} for short.

\begin{definition}\label{def.KolmK}
    The {\em conditional prefix Kolmogorov complexity} of $x$ given $y$ (as
auxiliary information) {\em with respect to prefix Turing machine} $T_i$ is
                  \begin{equation}\label{eq.KC}
    K_i(x | y) = \min_p \{|p|: \phi_i(p,y)=x \}.
                  \end{equation}
The {\em conditional prefix Kolmogorov complexity} $K(x | y)$ is defined
as the conditional Kolmogorov complexity
$K_u (x | y)$ with respect to the reference prefix Turing machine $T_u$
usually denoted by $U$.
The {\em unconditional} version is set to  $K(x)=K(x  | \epsilon)$.
\end{definition}

The prefix Kolmogorov complexity $K(x| y)$ satisfies the following so-called Invariance
Theorem:
\begin{equation}\label{eq.it}
K(x| y)\le K_i(x | y)+c_i
\end{equation}
 for
all $i,x,y$, where $c_i$ depends only on
$i$ (asymptotically, the reference machine is not worse
than any other machine).
Intuitively, $K(x| y)$ represents the minimal amount of information
required to generate $x$ by any effective process from input $y$ (provided
the set of programs is prefix-free).
The functions $K( \cdot)$ and $K( \cdot |  \cdot)$,
though defined in terms of a
particular machine model, are machine-independent up to an additive
constant
 and acquire an asymptotically universal and absolute character
through Church's thesis, see for example \cite{LV08},
and from the ability of universal machines to
simulate one another and execute any effective process.

Quantitatively, $K(x) \leq |x|+2 \log |x|+O(1)$.
A prominent property of the prefix-freeness of $K(x)$ is
that we can interpret $2^{-K(x)}$
as a probability distribution since $K(x)$ is the length of
a shortest prefix-free program for $x$. By the fundamental
Kraft's inequality \cite{Kr49} (see for example \cite{CT91,LV08}) we know that
if $l_1 , l_2 , \ldots$ are the code-word lengths of a  prefix code,
then $\sum_x 2^{-l_x} \leq 1$. Hence,
\begin{equation}\label{eq.udconv}
\sum_x 2^{-K(x)} \leq 1.
\end{equation}
This leads to the notion
of universal distribution ${\bf m}(x)=2^{-K(x)}$ which we may
view as a rigorous form of Occam's razor. Namely, the probability
${\bf m}(x)$ is great if $x$ is simple ($K(x)$ is small like $K(x)=O(\log |x|)$)
and ${\bf m}(x)$ is small if $x$ is complex 
($K(x)$ is large like $K(x) \geq |x|$).

  The Kolmogorov complexity of an individual finite object was introduced by
Kolmogorov \cite{Ko65} as an absolute
and objective quantification of the amount of information in it.
The information theory of Shannon \cite{Sh48}, on the other hand,
deals with {\em average} information {\em to communicate}
objects produced by a {\em random source}.
  Since the former theory is much more precise, it is surprising that
analogs of theorems in information theory hold for
Kolmogorov complexity, be it in somewhat weaker form.
An example is the remarkable {\em symmetry of information} property
used later, see \cite{ZL70} for the plain complexity version,
and \cite{Ga74} for the prefix complexity version below.
Let $x^*$  denote the shortest prefix-free program $x^*$
for a finite string $x$,
or, if there are more than one of these, then $x^*$ is the first
one halting in a fixed standard enumeration of all halting programs.
Then, by definition, $K(x)=|x^*|$.
Denote $K(x,y)=K(\langle x,y \rangle)$. Then,
\begin{align}\label{eq.soi}
K(x,y) & = K(x)+K(y \mid x^*) + O(1) \\
& = K(y)+K(x \mid y^*)+O(1) .
\nonumber
\end{align}
\begin{remark}\label{rem.xKx}
\rm
The information contained in $x^*$ in the conditional above
is the same as the information in the pair $(x,K(x))$, up to an additive
constant, since there are recursive functions $f$ and $g$ such that
for all $x$ we have $f(x^*)=(x,K(x))$ and $g(x,K(x))=x^*$.
On input $x^*$,  the function $f$ computes $x=U(x^*)$ and $K(x)=|x^*|$;
and on input $x,K(x)$ the function $g$ runs all programs of length $K(x)$
simultaneously, round-robin fashion, until the first program computing
$x$ halts---this is by definition $x^*$.
\end{remark}

\subsection{Computability Notions}\label{app.complexity}

If a function has as values pairs of nonnegative integers,
such as $(a,b)$, then we can interpret this value as the rational $a/b$.
We assume the notion of a computable function with rational arguments
and values.
A function $f(x)$ with $x$ rational is \emph{semicomputable from below}
if it is defined by a rational-valued total computable function $\phi(x,k)$
 with $x$ a rational number
and $k$ a nonnegative integer
such that $\phi(x,k+1) \geq \phi(x,k)$ for every $k$ and
  $\lim_{k \rightarrow \infty} \phi (x,k)=f(x)$.
This means
that $f$ (with possibly real values)
 can be computably approximated arbitrary close from below
 (see \cite{LV08}, p. 35).  A function $f$ is  \emph{semicomputable
from above} if $-f$ is semicomputable from below.
 If a function is both semicomputable from below
and semicomputable from above then it is \emph{computable}.

We now consider a subclass of the lower semicomputable functions.
A function $f$ is a {\em semiprobability} mass function if
$\sum_x f(x) \leq 1$ and it is a {\em probability} mass function
if $\sum_x f(x) = 1$. It is customary to write $p(x)$ for
$f(x)$ if the function involved is a semiprobability mass function.

\subsection{Precision}\label{app.precision}
It is customary in this area to use ``additive constant $c$'' or
equivalently ``additive $O(1)$ term'' to mean a constant,
accounting for the length of a fixed binary program,
independent from every variable or parameter in the expression
in which it occurs. In this paper we use the prefix complexity
variant of Kolmogorov complexity for convenience.
Prefix
complexity of a string exceeds
the plain complexity of that string by at most an
additive term that is logarithmic in the length of that string.

\section*{Acknowledgements}
We thank Alexander Shen and P\'eter G\'acs for commenting on the manuscript.

\end{document}